\documentclass[hidelinks, 12pt]{article}
\usepackage[mathscr]{eucal}
\usepackage{amssymb,latexsym,amsthm,pifont,graphicx,anysize,multicol}
\usepackage{verbatim}
\usepackage{graphicx}
\usepackage{color}
\usepackage{amsmath}
\usepackage{amsthm}
\usepackage{enumerate}
\usepackage{framed}
\usepackage{authblk}
\usepackage{setspace}
\usepackage{hyperref}
\usepackage{cancel}

\usepackage{bm}
\usepackage{tocvsec2}
\usepackage{comment} 
\usepackage{setspace}
\usepackage{hyperref}
\usepackage{epigraph}
\usepackage[margin=1.5cm]{caption}

\definecolor{dgray}{RGB}{90,90,90}
\definecolor{gray}{RGB}{120,120,120}
\definecolor{lgray}{RGB}{150,150,150}
\definecolor{purple}{RGB}{150,0,150}
\definecolor{magenta}{RGB}{250,0,150}

\numberwithin{equation}{section}

\newtheorem{thm}{Theorem}[section]
\newtheorem{lem}[thm]{Lemma}

\newtheorem{conj}[thm]{Conjecture}

\renewcommand\S{\Sigma}
\newcommand\s{\sigma}

\newcommand\D{\nabla}
\newcommand\e{\epsilon}
\renewcommand\b{\beta}

\newcommand\g{\gamma}

\renewcommand\a{\alpha}

\newcommand{\field}[1]{\mathbb{#1}}

\DeclareFontFamily{OT1}{rsfs}{}
\DeclareFontShape{OT1}{rsfs}{m}{n}{ <-7> rsfs5 <7-10> rsfs7 <10->
rsfs10}{} \DeclareMathAlphabet{\mycal}{OT1}{rsfs}{m}{n}

\newcommand\beq{\begin{equation}}
\newcommand\eeq{\end{equation}}
\newcommand\ben{\begin{enumerate}}
\newcommand\een{\end{enumerate}}
\newcommand\bit{\begin{itemize}}
\newcommand\eit{\end{itemize}}

\newcounter{mnotecount}[section]

\title{Rigidity in vacuum under conformal symmetry}

\author[*]{Gregory J. Galloway}
\author[$\dag$]{Carlos Vega}

\affil[*]{\small Department of Mathematics, 

University of Miami, Coral Gables, FL }

\affil[$\dag$]{Department of Mathematics,

Binghamton University SUNY, Binghamton, NY}

\begin{document}
\date{}

\newpage
\thispagestyle{empty}
\maketitle

\begin{abstract} 
Moitvated in part by \cite{IsenbergMoncrief}, in this note we obtain a rigidity result for globally hyperbolic vacuum spacetimes in arbitrary dimension that admit a timelike conformal Killing vector field.  Specifically, we show that if $M$ is a Ricci flat, timelike geodesically complete spacetime with compact Cauchy surfaces that admits a timelike conformal Killing field $X$, then $M$ must split as a metric product, and $X$ must be Killing. This gives a partial proof of the Bartnik splitting conjecture in the vacuum setting.  
\end{abstract}

\section{Introduction} 

The classical Hawking-Penrose singularity theorems form a cornerstone in the global theory of spacetime geometry and general relativity.  These theorems guarantee the existence of incomplete causal geodesics, (i.e., `singularities'), in large, generic classes of spacetimes satisfying natural energy conditions. 

The singularity theorems can be viewed as Lorentzian analogs to Riemannian Ricci comparison theorems like Myers' Theorem, and rely on \emph{strict} curvature conditions. In the early 1980's, S.-T. Yau put forth the question of the \emph{rigidity} of the singularity theorems, and posed a Lorentzian analog to the Cheeger-Gromoll splitting theorem in his famous problem section, \cite{Yau}, in 1982. This was settled in a series of papers by the end of the decade, \cite{Esch}, \cite{GalLST}, \cite{NewmanLST}, with the basic version of the result (due to Eschenburg) as follows:

\begin{thm} [Lorentzian Splitting Theorem] \label{LST}
Let $M$ be a globally hyperbolic, timelike geodesically complete spacetime, satisfying the timelike convergence condition, $\mathrm{Ric}(X,X) \ge 0$, for all timelike $X$. If $M$ admits a timelike line, then $M$ splits as an isometric product
\beq
(M^{n+1}, g) \approx (\field{R} \times \S^n, -dt^2 + h)\label{splitspacetime}
\eeq
where $\S^n$ is a smooth, geodesically complete, spacelike (Cauchy) hypersurface, with induced metric $h$.
\end{thm}

Despite the resolution of Theorem \ref{LST}, however, the result did not ultimately settle the original motivating rigidity question. A concrete formulation of this was posed by Bartnik in 1988 as follows:

\begin{conj} [Bartnik Splitting Conjecture] \label{BC} Let $M$ be a spacetime with compact Cauchy surfaces, which satisfies $\mathrm{Ric}(X,X) \ge 0$ for all timelike $X$. If $M$ is timelike geodesically complete, then $M$ splits as in \eqref{splitspacetime}, (with $\S$ compact).
\end{conj}

The conjecture has been established under various auxiliary conditions; see for example \cite{Gal84}, \cite{Bart88}, \cite{EschGal}, \cite{GalBanach}, \cite{horo1}, \cite{horo2}. It was proven under the stronger sectional curvature condition in \cite{EhrGal}, using \cite{Harris}.

Conjecture \ref{BC} is most simply illustrated by the special case of a warped product $(M^{n+1},g) = (I \times \S^n, -d\tau^2 + \phi^2(\tau)\widetilde{h})$, where $(\S^n,\widetilde{h})$ is a compact Riemannian manifold, $I \subset \field{R}$ an open interval, and $\phi : I \to (0,\infty)$ a smooth, positive funtion. For such a spacetime, the timelike convergence condition forces $\phi'' \le 0$. But then if $M$ is timelike complete, we must have $I = \field{R}$, which forces $\phi =c$ to be constant. (Then $h:= c^2\tilde{h}$ is the induced metric on $\S$, and hence $M$ splits as above.) 

While the warped product case is trivial, one may ask what happens when this is `weakened' to the assumption of a timelike conformal symmetry, i.e., the existence of a timelike conformal Killing field. By the latter, we mean a timelike vector field $X$ such that $\mathcal{L}_Xg = 2\sigma g$, where $\mathcal{L}$ is the Lie derivative, and $\s : M \to \field{R}$ is smooth. For example, a warped product as above has timelike conformal Killing vector field $X = \phi(\tau)\partial_\tau$, with conformal factor $\sigma = \phi'(\tau)$. In \cite{IsenbergMoncrief}, various results are established showing the `rigidity' imposed by the existence of conformal symmetries on solutions $(M^4, g)$ of the Einstein equations. Theorem 3 in \cite{IsenbergMoncrief}, for example, shows that a vacuum solution with a proper conformal symmetry must be one of a few special types. The proof makes special use of the dimension $3+1$.

Interestingly, Conjecture \ref{BC} remains open even in the vacuum setting, $\textrm{Ric} \equiv 0$. Indeed, we are not aware of any prior results in this direction. The main result established here is the following:

\begin{thm}\label{mainsplit} 
Let $n \ge 2$ and suppose that $(M^{n+1},g)$ is a Ricci flat, timelike geodesically complete spacetime, with compact Cauchy surfaces. If $M$ admits a timelike conformal Killing field $X$, then $M$ splits isometrically as
\beq
(M^{n+1}, g) \approx (\field{R} \times \S^n, -dt^2 + h)\label{splitspacetime2}
\eeq
with the (Riemannian) fiber $(\S^n, h)$ compact and Ricci flat, and $X$ is in fact Killing. 
\end{thm}

\section{Preliminary results}

Recall that a smooth vector field $X$ on a semi-Riemannian manifold $(M,g)$ is \emph{Killing} if $g$ is invariant under the flow of $X$, i.e., if $\mathcal{L}_X g = 0$, where $\mathcal{L}$ is the Lie derivative. More generally, by a \emph{conformal Killing field} on $(M,g)$, we mean a smooth vector field $X$, such that $\mathcal{L}_X g = 2\sigma g$, for some smooth function $\sigma : M \to \field{R}$. In the special case that $\sigma$ is a constant, $X$ is called a \emph{homothetic Killing field}.

We begin with the following standard observation. (The notation is suggestive for our applications below, but note that we are not assuming that $X$ is timelike or that $M$ is  Lorentzian.)

\begin{lem} \label{conformalcoords} Let $M = (M^{n+1},g)$ be a semi-Riemannian manifold, and suppose that $X$ is a conformal Killing field, with $\mathcal{L}_X g = 2\sigma g$. Then, in any local coordinates $\{t = x^0, x^1, ..., x^n\}$ with $X = \partial_t$, we have:
\beq
g(t,x^1, ..., x^n) = e^{2f(t,x^1, ... , x^n)} \displaystyle \sum_{i,j = 0}^n G_{ij}(x^1, ..., x^n)dx^i \otimes dx^j \nonumber
\eeq
where $f(t,x^1, ..., x^n) := \int_0^{t} \sigma(s, x^1, ... , x^n) d s$. In particular, note that $\partial_t f = \sigma$. 
\end{lem}
\begin{proof} Because $\mathcal{L}_X(\partial_i) = [\partial_t, \partial_i] = 0$, we have $2\sigma g_{ij} = (2\sigma g)(\partial_i, \partial_j) = (\mathcal{L}_X g)(\partial_i, \partial_j) = X [g(\partial_i, \partial_j)] = Xg_{ij} = \partial_t g_{ij}$
Hence, for any indices $i, j \in \{0, 1, 2, ..., n\}$, we have $\partial_t g_{ij} = 2 \sigma g_{ij}$. Both sides are functions of $(t = x^0, x^1, ..., x^n)$, but holding $x^i$ constant for all $i \ge 1$, we have a first order linear equation in the single variable $t$. Using the integrating factor $\mu = e^{-2f}$ gives $g_{ij}(t, x^1,..., x^n) = e^{2f(t, x^1,..., x^n)}G_{ij}(x^1, ..., x^n)$.
\end{proof}

We note that for a warped product spacetime metric $g = -d\tau^2 + \phi^2(\tau)\tilde{h}$, the above result holds globally.   For example, letting $t := \int_c^\tau 1/\phi(s)ds$, then $d t = d\tau/\phi(\tau)$, and 
\beq
g = -d\tau^2 + \phi^2(\tau)\tilde{h} = \phi^2(\tau(t))(-dt^2 + \tilde{h}) = e^{2f(t)}(-dt^2 + \tilde{h})\nonumber
\eeq
Indeed, this is precisely the form of the metric induced as in Lemma \ref{conformalcoords} by the conformal Killing field $X = \phi(\tau)\partial_\tau = \partial_t$, with $df/dt = d\phi/d\tau = \sigma$.

\vspace{1pc}
We shall make use of the following:
\begin{lem} \label{conformalconstantofmotion} Let $M$ be a semi-Riemannian manifold, and let $X$ be a conformal Killing field, with $\mathcal{L}_Xg = 2\s g$. Then we have:
\ben
\item [(1)] $g(\nabla_YX,Y) = \sigma g(Y,Y)$, for all smooth vector fields $Y$.
\item [(2)] Let $\g = \g(s)$ be any affinely parameterized geodesic, and set $C := g(\g'(s), \g'(s))$. Then along the geodesic $\g = \g(s)$, we have:
\beq
\frac{d}{ds} g(X,\g') = \s(\g(s)) C\label{conformalconstanteq}
\eeq
\een
\end{lem}
\begin{proof} (1) follows from the standard formula:
\beq
(\mathcal{L}_Xg)(V,W) = g(\nabla_VX,W) +g(\nabla_WX,V)\nonumber \label{LieNabla}
\eeq
To prove (2), note that for any curve $\g$, we have:
\beq
\g' g(X,\g') = g(\nabla_{\g'}X, \g') + g(X, \nabla_{\g'}\g')\nonumber\label{}
\eeq
If $\g$ is a geodesic, the last term vanishes, and the result follows from (1).
\end{proof}

\vspace{1pc}
The following lemma is the key analytic result needed to prove Theorem \ref{mainsplit}.

\begin{lem} \label{CKVFEinstein} For $n \ge 2$, suppose that $(M^{n+1},g)$ is a semi-Riemannian manifold, with $\mathrm{Ric}_g = \lambda g$, for some real number $\lambda \in \field{R}$, and suppose that $X$ is a nowhere vanishing conformal Killing field, with $\mathcal{L}_X g = 2\sigma g$. Then, with $\Delta_g \sigma = \mathrm{tr} (\mathrm{Hess}_g(\sigma))$, we have:

\vspace{.5pc}\label{hesseqn}
\beq
\mathrm{Hess}_g(\sigma) = - \bigg(\frac{\Delta_g \sigma + 2\lambda \sigma}{n-1}\bigg)g\label{CKVFEinsteineq1}
\eeq

\vspace{.5pc}
\noindent
which after tracing gives:

\beq
\Delta_g \sigma = -\lambda \bigg(\frac{n+1}{n} \bigg)\sigma   \label{CKVFEinsteineq2}
\eeq

\end{lem}

\smallskip
\begin{proof} For the convenience of the reader we provide an outline of the proof, which is a lengthy computation. (See also the proof of Theorem 3 in \cite{IsenbergMoncrief}, which treats $\lambda = 0$, citing \cite{Yano} for the relevant formula in this case.) 

Fix local coordinates $\{t = x^0, x^1, ..., x^n\}$, and $f$ and $G$, as in Lemma \ref{conformalcoords}, with $X = \partial_t$. Hence, in the neighborhood $U$ covered by the chart, $g$ is conformal to a metric whose components are independent of $t$, that is, on $U$ we have $g = e^{2f}G$, with $G = G(x^1,...,x^n)$. Using the formula for Ricci under conformal change in \cite{Besse}, we have:
\beq
\mathrm{Ric}_g = \mathrm{Ric}_G - (n-1) \bigg( H_G^f - df \otimes df \bigg) - \bigg( \Delta_G f + (n-1)||df||_G^2 \bigg) G \label{Besse}
\eeq
where $||df||^2_G = G(\nabla_Gf, \nabla_Gf)$, $H_G^f = \textrm{Hess}_G(f)$ and $\Delta_G f = \mathrm{tr} (H_G^f)$. (Note: The sign convention in \cite{Besse} is $\Delta_G f := - \mathrm{tr} (H_G^f)$.) Applying the Einstein condition we obtain:
\beq
\bigg(\frac{\lambda }{n-1}\bigg)g  = \frac{\mathrm{Ric}_G}{n-1} - H_G^f + df \otimes df - ||df||_G^2 G - \bigg(\frac{\Delta_G f}{n-1}\bigg)G \label{vcR}  \,.
\eeq

We now take the Lie derivative of \eqref{vcR} with respect to $X = \partial_t$. First note that $\mathcal{L}_X(\mathrm{Ric}_G) = \mathcal{L}_{\partial t} (\mathrm{Ric}_G)$ vanishes, since the coefficients 
$(\mathrm{Ric}_G)_{ij}$, which depend only on $G_{ij}$ and its derivatives, are independent of $t$.
Thus, taking the Lie derivative of \eqref{vcR} gives:
\beq
\bigg(\dfrac{2\lambda \sigma}{n-1}\bigg)g = - \mathcal{L}_{\partial_t} H_G^f  + \mathcal{L}_{\partial_t} (df \otimes df) - \mathcal{L}_{\partial_t} \bigg(||df||_G^2 G \bigg) - \mathcal{L}_{\partial_t} \bigg( \frac{\Delta_G f}{n-1}G \bigg) \label{vcR2}
\eeq
One may now proceed to compute the four Lie derivatives in \eqref{vcR2}, using, where appropriate, the fact that the  $G_{ij}$'s, and quantities defined in terms of the $G_{ij}$'s, have vanishing $t$-derivative, and $\partial_t f = \s$.  One obtains,

\beq
\mathcal{L}_{\partial_t} H_G^f  =  H_G^\sigma \,, \quad  \mathcal{L}_{\partial_t} [(\Delta_G f )G] 
= (\Delta_G \sigma) G \nonumber
\eeq
\beq
\mathcal{L}_{\partial_t} (df \otimes df) = d\sigma \otimes df + df \otimes d\sigma \, , \quad   \partial_t ||df||_G^2 =
2G(\nabla_G \sigma, \nabla_G f)  \,.  \nonumber
\eeq
Substituting these into \eqref{vcR2} gives:
\beq
\bigg(\dfrac{2\lambda \sigma}{n-1}\bigg)g = - H_G^\sigma + d\sigma \otimes df + df \otimes d\sigma - 2G(\nabla_G \sigma, \nabla_G f)G - \bigg( \frac{\Delta_G \sigma}{n-1} \bigg) G \label{vcR3}
\eeq
We now translate all the $G$-terms in \eqref{vcR3} back to the metric $g$. First note that:
\beq
G(\nabla_G \sigma, \nabla_G f)G  = g(\nabla_g \sigma, \nabla_g f)g  \nonumber 
\eeq
By standard formulas, we have for the Hessian and Laplacian,
\begin{align*}
H_G^\sigma &= H_g^\sigma + d\sigma \otimes df + df \otimes d\sigma - g(\nabla_g \sigma, \nabla_g f) g  \\
 \Delta_G \sigma & =   e^{2f} \bigg( \Delta_g \sigma - (n-1) g(\nabla_g \sigma, \nabla_g f) \bigg) \,
\end{align*}
By plugging these pieces into \eqref{vcR3}, after some simple manipulations we arrive at \eqref{CKVFEinsteineq1} and \eqref{CKVFEinsteineq2}.
\end{proof}

\vspace{1pc}
The proof of Theorem \ref{mainsplit} eventually reduces to the static case. We will then make use of the following curve lifting result.

\begin{lem} \label{Killinglifts} Let $M$ be a globally hyperbolic spacetime, with smooth spacelike Cauchy surface $S$. If $M$ admits a complete timelike Killing field $X$, then every spatial curve in $S$ lifts (along the integral curves of $X$) to a timelike curve in $M$. 
\end{lem}

\begin{proof} Because $X$ is complete, we have a diffeomorphic splitting $M \approx \field{R} \times S$, given by flowing along the integral curves of $X$. By reparameterizing if necessary, we may suppose that each integral curve $\g(t)$ of $X$ meets $S$ at $t = 0$. 

We will now prepare a convenient collection of coordinate patches on $M^{n+1}$. First note that, choosing any local coordinates $\{x^1, ..., x^n\}$ on $S$, then $\{t = x^0, x^1, ..., x^n\}$ give local coordinates on $M$, and by Lemma \ref{conformalcoords} we have:
\beq
g(t,x^1, ..., x^n) = \displaystyle \sum_{i,j = 0}^n G_{ij}(x^1, ..., x^n)dx^i \otimes dx^j \label{Killingmetric}
\eeq
For $p \in S$, let $U_p$ be a neighborhood of $p$ in $S$, with local coordinates $\{x^1, ..., x^n\}$. Let $V_p$ be a smaller neighborhood, with $p \in V_p \subset \subset U_p$. Hence, $\{t = x^0, x^1, ..., x^n\}$ give local coordinates on $\field{R} \times V_p$, on which $g$ has coordinate representation \eqref{Killingmetric}. Because $g$ is continuous, and $V_p \subset \subset U_p$, and because the component functions $G_{ij}$ are independent of $t$, the $G_{ij}$ are bounded on $\field{R} \times V_p$. Moreover, because $G_{00}$ is negative on $\field{R} \times V_p$, we have 
\beq
m_p := \min \{-G_{00}(z) : z \in \field{R} \times V_p \} > 0
\eeq

Now fix any spatial curve $\b : [0,\ell] \to S$. Since the image of $\b$ is compact, we can find finitely many points $\{p_1, ..., p_N\}$ such that $\textrm{Im}(\b) \subset (V_{p_1} \cup \cdots \cup V_{p_N})$. Hence, $\b$ breaks into finitely many subsegments, with each contained in a single patch. Provided that we are able to lift to any desired `initial height' or `starting time', it thus suffices to assume that $\b$ lies in a single chart as above, say, $\textrm{Im}(\beta) \subset V_p$. Hence, we have $\beta(u) = (\beta_1(u), ..., \beta_n(u))$, for $u \in [0,\ell]$. For $0< s \le1$, consider $\beta_s(u) = \beta(su)$, for $u \in [0, \ell/s]$. Consider the simple lift up to the starting time $t = t_0$, given by $\alpha_s(t) = (t + t_0, \beta(st))$. Then, for $t \in [0, \ell/s]$, 
\begin{eqnarray}
g(\alpha'_s(t), \alpha'_s(t)) & =  & G_{00}(\b(st)) + 2sG_{0i}(\b(st))\beta'_i(st) + s^2G_{ij}(\b(st))\beta'_i(st) \beta'_j(st) \nonumber \\[.5pc]
  & \le & \; \; \; \;  -m_p \; \; \; \; \; + 2sG_{0i}(\b(st))\beta'_i(st) + s^2G_{ij}(\b(st))\beta'_i(st) \beta'_j(st)\nonumber
\end{eqnarray}
\noindent
where we sum over all repeated indices, with $i,j \in \{1, ..., n\}$. Since $- m_p$ is strictly negative and everything else is bounded, we can find an $s$ small enough so that this last quantity is negative, and hence so that $\alpha_s(t)$ is timelike, for all $t \in [0, \ell/s]$. \end{proof}

\section{Proof of the splitting result}

\proof[Proof of Theorem \ref{mainsplit}]
Applying Lemma \ref{CKVFEinstein} with $\lambda = 0$, we see that  $\nabla \sigma$ is parallel. Fix any smooth, spacelike Cauchy surface, $S$. By compactness, $\sigma|_S$ attains a maximum at some $p \in S$, and thus $(\nabla \sigma)_p$ is normal to $S$. Thus, either $(\nabla \sigma)_p$ is timelike or zero. But since $\nabla \sigma$ is parallel, then either $\nabla \sigma$ is everywhere timelike, or it vanishes identically. That is, either $\nabla \sigma$ is an everywhere timelike vector field, or $\sigma$ is constant. The proof will show that, in fact, $\s$ must be zero, but we proceed by considering each case below.

Case 1: Suppose first that $\nabla \sigma$ is everywhere timelike. By reversing the time-orientation of $M$ if necessary, we may suppose that $\nabla \sigma$ is future-pointing. Since $\nabla \sigma$ is parallel, its integral curves are timelike geodesics, and hence complete by assumption. 
Moreover, the quantity $g(\D\s,\D\s)$ is constant, which by a rescaling, can be taken to be $-1$. The condition $g(\nabla \s, \nabla \s) = -1$ then forces the integral curves of $\nabla \s$ to be maximal (and hence to be timelike lines). To see this, fix one such integral curve, $\g$. Since $\nabla \s$ is future-pointing, $\g$ is a future-directed, unit-speed timelike geodesic. Without loss of generality, suppose that $\a : [a,b] \to M$ is another future timelike curve from $\a(a) = \g(0)$ to $\a(b) = \g(\ell)$. Since $-1 = g(\nabla \s, \nabla \s) = g(\nabla \s, \g') = (\s \circ \g)'$, note that along $\g$ we have $\s(\g(u)) = -u + \s(\g(0))$. Also, because $\nabla \s$ is future-pointing, we have $g(\a', \nabla \s) < 0$.
Define the `spacelike part' of $\a'$ by $N := \a' + g(\a',\nabla \s)\nabla \s$. It follows that $N$ is a vector field on $\a$ with $g(N,N) \ge 0$ and $g(N, \nabla \s) = 0$, and $\a' = -g(\a', \nabla \s)\nabla \s + N$. Then $|\a'| = [-g(\a', \a')]^{1/2} = [g(\a', \nabla \s)^2 - g(N,N)]^{1/2} \le -g(\a', \nabla \s) = - \a'(\s) = - \frac{d}{ds}(\s (\a(s))$. Integrating this gives $L(\a) = \int_a^b|\a'|ds \le \s(\a(a))-\s(\a(b)) = \s(\g(0)) - \s(\g(\ell)) = \ell = L(\g|_{[0,\ell]})$. This shows that the (arbitrary) subsegment $\g|_{[0,\ell]}$ is maximal, and thus $\g$ is a timelike line. (A local version of this basic maximality argument appears, for example, in Proposition 34 in Chapter 5 of \cite{ON}.)

It now follows that $M$ splits as a product, as in \eqref{splitspacetime}, with compact, totally geodesic spacelike slices $\{t\}\times \S$. Since $\nabla \s$ is parallel, it follows from a standard maximum principle argument that each level set of $\s$ must coincide with a slice in the splitting. Fix any nonzero level set $\{\s = k\}$, $k \ne 0$. Since $\{\s = k\}$ is a slice in the product, it is totally geodesic and compact, and hence admits a closed spacelike geodesic, $\g$. But then \eqref{conformalconstanteq} in Lemma \ref{conformalconstantofmotion} leads to a contradiction as we traverse a full circuit of $\g$.

Case 2: We have shown that $\sigma$ must in fact be constant, i.e., that $X$ is in fact a homothetic Killing field, $\mathcal{L}_X g = 2cg$, for some constant $c$. We now claim that $c = 0$ and $X$ is Killing. To see this, let $\alpha : \field{R} \to M$ be a complete unit speed timelike geodesic.  Then we have $-c  =  cg(\alpha', \alpha') = g(\nabla_{\alpha'}X,\alpha') = \alpha'(g(X,\alpha'))$. This implies $g(X,\alpha'(s)) = -cs+d$. If $c \not = 0$, then $g(X,\a'(d/c)) = 0$ gives a contradiction, since both $\a'$ and $X$ are timelike.

Hence $X$ is Killing. The assumption of timelike completeness implies, in fact,  that $X$ is a complete Killing vector field; cf. \cite[Lemma 1]{GarHar}.
Without loss of generality, we may assume that $X$ is future-pointing. For $a \in S$, we can think of the integral curve $L_a$ of $X$ as the `spatial location' corresponding to $a \in S$, flowing in time. Since $X$ is complete and Killing, it follows from Lemma \ref{Killinglifts} that any spatial curve in $S$ lifts, along the integral curves of $X$, to a timelike curve in $M$. It follows that for each $a, b \in S$, there is some finite time $t \in (0,\infty)$ for which $L_b(t) \in I^+(a)$. For $a, b \in S$, define the shortest such `commute time' from $a$ to $b$, (really from $a$ to $L_b$) by:
\beq
C_a(b) := \inf \{ t : L_b(t) \in I^+(a)\} \label{commute}
\eeq
Letting $C(a,b) = C_a(b)$, we claim that $C : S \times S \to [0, \infty)$ is upper-semicontinuous. Suppose so for the moment. Since $S$ is compact, it then follows that $C$ is bounded, that is, there is a `maximum commute time' $\tau$, such that $L_b(\tau) \in I^+(a)$, for all $a,b \in S$. Let $\g : [0, \infty) \to M$ be any future timelike unit-speed $S$-ray. (Hence, $\g$ is future-inextendible, with $d(S, \g(s)) = s$ for all $s \ge 0$.) Note that the integral curves of $X$ give a diffeomorphic splitting $M \approx \field{R} \times S$. Since the slab $[0, \tau] \times S$ is compact, $\g$ must meet the slice $\{\tau\} \times S$, at some point $\g(s_0) = (\tau, a_0) = L_{a_0}(\tau)$. It follows that $S \subset I^-(\g(s_0)) \subset I^-(\g)$. But then by Theorem A in \cite{EschGal}, $M$ contains a timelike line.  Thus, again, by the Lorentzian splitting theorem, $M$ splits. 

It remains to show that `commute time' $C : S \times S \to [0, \infty)$ is upper semicontinous. Fix $a, b \in S$, and $\e >0$. Let $t_0 := C_a(b)$, and set $t_1 := t_0 + \e/2$. Hence, $L_b(t_1) \in I^+(a)$. Letting $\pi_S : M \to S$ be the standard projection, $\pi_S(L_x(t)) = x$,  define $U_a := \pi_S(I^-(L_b(t_1)))$. Hence, $U_a$ is an open neighborhood of $a$ in $S$, such that $L_b(t_1) \in I^+(x)$, for all $x \in U_a$. Letting $t_2 := t_0 + \e$ and $S_{t_2} := \{t_2\} \times S$, define $V_b := \pi_S(I^+(L_b(t_1)) \cap S_{t_2})$. Hence, $V_b$ is a neighborhood of $b$ in $S$ such that $L_y(t_2) \in I^+(L_b(t_1))$, for all $y \in V_b$. But then, for all $(x,y) \in U_a \times V_b \subset S \times S$, we have $L_y(t_2) \in I^+(L_b(t_1)) \subset I^+(x)$, i.e., $L_y(t_0+\e) \in I^+(x)$. In other words, for all $(x,y) \in U_a \times V_b$, we have $C_x(y) \le t_0 + \e = C_a(b) + \e$. 

We have now shown that $M$ splits as in \eqref{splitspacetime2}, for some smooth, compact spacelike Cauchy hypersurface $\S^n$. That this, with its induced Riemannian metric $h$, has $\textrm{Ric}_h = 0$ follows, for example, from the warped product curvature formulas in \cite{ON}, with $\textrm{Ric}_g = 0$ and $f = 1$.\qed

\bigskip
\noindent
\textsc{Acknowledgements:}  GJG's research was supported in part by
NSF grants DMS-1313724 and DMS-1710808.

\bibliographystyle{amsplain}
\bibliography{ckvf}  

\end{document}